\documentclass[12pt]{amsart}
\usepackage{amsmath}
\usepackage{amssymb}
\usepackage{mathscinet}
\usepackage{amsfonts}
\usepackage{epsfig}
\usepackage{amsthm}
\usepackage{latexsym}
\usepackage[latin1]{inputenc}
\usepackage{verbatim}
\usepackage{pb-diagram}
\usepackage[small]{caption} 
\usepackage[labelformat=empty]{subfig}
\usepackage{everysel, keyval, ragged2e}
\usepackage{wrapfig}

\newtheorem{theorem}{Theorem}[section]
\newtheorem{lemma}[theorem]{Lemma}
\newtheorem{corollary}[theorem]{Corollary}
\newtheorem{proposition}[theorem]{Proposition}
\theoremstyle{definition}
\newtheorem{definition}[theorem]{Definition}

\theoremstyle{remark}

\newtheorem{observation}[theorem]{Observation}
\numberwithin{equation}{section}

\newcommand{\m}{\mathcal }

\newcommand{\ket}[1]{| #1 \rangle}    
\def\R{\mathbb R}    

\def\N{\mathbb N}    

\begin{document}

\title{Dirichlet polynomials, Majorization, and Trumping}

\author[Rajesh Pereira and Sarah Plosker]{Rajesh Pereira and Sarah Plosker}
\address{Department of Mathematics \& Statistics, University of Guelph, Guelph, ON, Canada N1G 2W1}
\email{pereirar@uoguelph.ca, splosker@uoguelph.ca}

\begin{abstract}
Majorization and trumping are two partial orders which have proved useful in quantum information theory.  We show some relations between these two partial orders and generalized Dirichlet polynomials, Mellin transforms, and completely monotone functions.  These relations are used to prove a succinct generalization of Turgut's characterization of trumping.

\end{abstract}


\maketitle
\section{Introduction}

Majorization is a partial order on
real vectors that is used primarily in matrix analysis, although historically it has been explored in economics as a means of  studying inequalities of wealth and inequalities of income \cite{MOA11}.

Entanglement is the crux of quantum information theory; being able to manipulate entanglement is what makes quantum information such a powerful tool. A major area of research in quantum information theory is the problem of entanglement transformations: that is, can we manipulate a pure state of a composite system via local operations and classical communication (LOCC) and have it transform into another particular state? Recently, Nielsen \cite{Nie99} answered this question using    majorization theory, thus giving
majorization  an important role in quantum information theory.

If two parties, Alice and Bob, can  only carry out operations on their local systems and have a
classical communication channel to transmit bits, it is called \emph{local operations and classical communication (LOCC)} \cite{Betal99}.
The most general implementation of LOCC is a potentially unlimited back-and-forth scenario:  (1) Alice applies a quantum channel to her system and communicates classical information to Bob, and  (2) Bob applies a quantum channel on his system and communicates classical information to Alice. However, in  \cite{LoPr01} the authors prove that all LOCC protocols can be simplified into a one-way communication setup. Physically, it follows from this theorem  that an arbitrary protocol transforming the state $\ket{\psi^A}$ to the state $\ket{\psi^B}$ using local operations and two-way classical communication can be simulated by a one-way communication protocol from Alice to Bob.

Suppose $|\psi\rangle$ is a pure state in a composite system $ \m H_A\otimes \m H_B$.  Using the language and notation set out in \cite{Nie99}, we denote by $\rho_\psi \equiv \mbox{tr}_B(|\psi\rangle \langle
\psi|)$ the state of Alice's system, and  $\lambda_\psi$ the vector of eigenvalues of
$\rho_\psi$.   We say that $|\psi\rangle \rightarrow |\phi\rangle$, read ``$|\psi\rangle$
transforms to $|\phi\rangle$'' if $|\psi\rangle$ can be
transformed into $|\phi\rangle$ by local operations and potentially
unlimited two-way classical communication.  Then we have:

\begin{theorem} \cite{Nie99} We can transform $|\psi\rangle$  to $|\phi\rangle$ using local
operations and classical communication if and only if $\lambda_\psi$
is majorized by $\lambda_\phi$.  That is,
\begin{eqnarray}
|\psi\rangle \rightarrow |\phi\rangle \mbox{ iff } \lambda_{\psi} \prec
\lambda_{\phi}.
\end{eqnarray}
\end{theorem}
This theorem is significant because any entangled state can be transformed via LOCC into a state that is less or equally entangled. Thus this theorem states that $|\psi\rangle$ is at least as entangled as $|\phi\rangle $ if and only if the eigenvalues of  $\rho_{\psi}$ are majorized by the eigenvalues of  $\rho_{\phi}$.

Because majorization is not a total order, there are many vectors that are \emph{incomparable} in that $x\not\prec y$ and $y\not\prec x$.
Jonathan and Plenio \cite{JoPl99} introduced the more general notion of \emph{trumping}, which  extends the number of comparable vectors. The idea being that, given two incomparable vectors $x$ and $y$, there sometimes exists a \emph{catalyst} $c$  (a real vector with positive components) such that $x\otimes c$ is majorized by $y\otimes c$.

A flurry of activity has ensued in order to build up the mathematical framework for trumping. In \cite{DaKl01}, the authors show that the dimension of the required catalyst is in general unbounded. They do this through careful manipulation of inequalities of summations of the components of the vectors in question. In \cite{Tur07},  Turgut  links trumping with $\ell^p$ inequalities and  the von Neumann entropy by manipulating differences of characteristic functions of the form $\sum_i(t-x_i)_+\equiv\sum_i\max\{(t-x_i), 0\}$ and making small variations to the components of the vectors, which do not affect the trumping relation.
In \cite{Kli2004,Kli2007}, Klimesh proves a result equivalent to that of Turgut using mathematical analysis.  In \cite{AuNe08}, Aubrun and Nechita describe the closure of the set of vectors trumped by a given vector $y$ via $\ell_p$ norms and describe trumping when considering infinite-dimensional catalysts. Klimesh, Turgut, and Aubrun and Nechita use completely different techniques for their proofs, exemplifying the fact that quantum information theory is multi-disciplined in nature, and problems can be solved using tools from many different fields of study.

Our contribution to the theory of trumping herein is to characterize trumping via the use of general Dirichlet polynomials, Mellin transforms, and completely monotone functions, thus approaching the subject from yet another point of view. Moreover, we generalize Turgut's result on trumping  \cite{Tur07} to convex sequences of higher order. The proof is simplified considerably by the use of the aforementioned mathematical concepts. 

The paper is organized as follows: in section \ref{sec:defnsnotn}, we recall the equivalent definitions of majorization and  the definition of trumping, which can be thought of as a generalized version of majorization, as well as the key theorem by Turgut \cite{Tur07} mentioned above. The remainder of section \ref{sec:defnsnotn} is spent reviewing Dirichlet polynomials, Mellin transforms, and completely monotone functions, which will all be used in subsequent sections.

In section \ref{sec:Connection}, we introduce an alternate definition of majorization based on linking the combined results from \cite{Mer05, Gav05} with the definition presented in \cite{KaNo63}. This connection appears to be new, and explicitly links majorization, a tool used in matrix analysis, to general Dirichlet polynomials, which arise in analytic number theory. We also show how this new definition of majorization  can be used to obtain a reinterpreted version of Turgut's theorem.

Finally, in section \ref{sec:r-convex}, we consider a generalization \cite{Nie05} of the work in  \cite{Mer05, Gav05} to convex sequences of higher order. Using the machinery of convex sequences and functions of higher order, we then present our main result: a generalization of Turgut's theorem. Our theorem reduces to that of Turgut's for $r=2$.

\section{Preliminary Definitions and Theorems}\label{sec:defnsnotn}

\subsection{Majorization}\label{sec:maj}
Let $x=(x_{1},x_{2},...,x_{d})\in \mathbb{R}^d$ and let
 $x^{\downarrow}=(x^{\downarrow}_{1},x^{\downarrow}_{2},...,x^{\downarrow}_{d})\in \mathbb{R}^d$  be the vector consisting of the elements of $x$ reordered so that \[x^{\downarrow}_{1}\ge x^{\downarrow}_{2}\ge x^{\downarrow}_{3}\ge...\ge x^{\downarrow}_{d}.\]

\begin{definition}\label{defn:maj}  Let $x=(x_{1},x_{2},...,x_{d}), y=(y_{1},y_{2},...,y_{d})\in \mathbb{R}^d$.
We say that $x$ is majorized  by $y$, written $x\prec y$,   if
\begin{eqnarray*}
\sum_{j=1}^{k}x^{\downarrow}_{j}\leq \sum_{j=1}^{k}y^{\downarrow}_{j}\quad 1\leq k \leq d,
\end{eqnarray*}
with equality when $k=d$.
\end{definition}

 We note the following well-known equivalent condition for majorization: $x\prec y$ if and only if  $\sum_{i=1}^d \Psi(x_i) \leq \sum_{i=1}^d \Psi(y_i)$, for every convex function $\Psi:\mathbb{R}\to \mathbb{R}$.

Now suppose $x_i,y_i\in \mathbb{N}$ for all $i:1\le i\le d$.  Let
\begin{eqnarray}\label{eq:a_n}
a_n=\# \{ i:y_i=n\} -\# \{ i:x_i=n\} \quad \forall n\in \mathbb{N}.
\end{eqnarray}
Then $\sum_{n=1}^{\infty} a_n\Psi(n)= \sum_{i=1}^d \Psi(y_i)- \sum_{i=1}^d \Psi(x_i)$ and hence $x\prec y$ if and only if $\sum_{n=1}^{\infty} a_n\Psi(n)\geq 0$ for all convex functions $\Psi:\mathbb{R}\to \mathbb{R}$ or equivalently $\sum_{n=1}^{\infty} a_n\phi(n)\geq 0$ for all convex sequences $\{ \phi(n)\}_{n=1}^{\infty} $.  This is a classical description of majorization.   We now consider general $d$-tuples of real numbers $\{ a_n\}$ for which $\sum_{n=1}^{\infty} a_n\Psi(n)\geq 0$ and do not restrict ourselves only to the integer values of $\{a_n\}$ which arise from equation \ref{eq:a_n}.

The following characterization of such sequences $\{ a_n\}$ was given in \cite{Gav05}, which built on the work of \cite{Mer05}.

\begin{theorem}\label{thm:Gav}
Let $a_0, \dots, a_d$ be $d+1$ fixed real numbers such that $\sum_{i=0}^da_i^2>0$. The inequality
\[
\sum_{i=0}^da_i\phi(i)\geq 0
\]
holds for every convex sequence $\{\phi(i)\}$ if and only if the polynomial $\sum_{i=0}^da_it^i$ has a double root at $t=1$ and all the coefficients of the polynomial
\begin{eqnarray}\label{eq:non-negcoeffs}
\frac{\sum_{i=0}^da_it^i}{(t-1)^2}=\sum_{i=0}^{d-2}c_it^i
\end{eqnarray}
are nonnegative.
\end{theorem}

As a simple example, consider the vectors $x=(5, 5, 5, 5), y=(2, 2, 6, 10)$. It can easily be verified that $x\prec y$. Consider $n\in\N$ for which $x_i=n$ or $y_i=n$ and let $a_n$ be defined as in equation (\ref{eq:a_n}). We have
\[
\sum_{\{n\,|\, x_i=n \textnormal{ or } y_i=n\}} a_nt^n=\sum_it^{y_i}-\sum_it^{x_i}=t^{10}+t^6-4t^5+2t^2.
\]
We can check that this polynomial has a double root at $t=1$ by dividing by $t^2-2t+1$ to obtain the polynomial $t^8+2t^7+3t^6+4t^5+6t^4+4t^3+2t^2$, which has no further roots at $t=1$. Furthermore, we can clearly see all the coefficients of this new polynomial are positive. By theorem \ref{thm:Gav}, we have $\sum_{0\leq i\leq m}a_i\phi(i)\geq 0$ for every convex sequence $\{\phi(i)\}$ and hence $x\prec y$.

We note that theorem \ref{thm:Gav} can be seen to be a reformulation of the following earlier result.

\begin{lemma}\cite{KaNo63} \label{lemma:KaNo63}
The inequality $\sum_{0\leq i\leq m} a_i\phi(i)\geq 0$ for all convex functions $\phi$  is equivalent to the following three conditions.
\begin{enumerate}
\item[(i)] $\sum_{0\leq i\leq m}a_i=0$;
\item[(ii)] $\sum_{0\leq i\leq m}ia_i=0$;
\item[(iii)] $\sum_{j=0}^k\sum_{i=0}^ja_i\geq 0 \quad 0\leq k\leq m$.
\end{enumerate}
\end{lemma}

To see that theorem \ref{thm:Gav} is equivalent to lemma  \ref{lemma:KaNo63}, we first note that condition (i) of the lemma is equivalent to the polynomial $\sum_ia_it^i$ having a root at $t=1$. Condition (ii) of the lemma is precisely the derivative of the polynomial $\sum_ia_it^i$ evaluated at $t=1$, so conditions (i) and (ii) are equivalent to the polynomial  $\sum_ia_it^i$ having a double root at $t=1$. Finally, re-writing equation (\ref{eq:non-negcoeffs}) in terms of the $a_i$, we find $a_i=c_i-2c_{i-1}+c_{i-2}$, where we let $c_{-1}=c_{-2}=0$. We then find  $\sum_{j=0}^k\sum_{i=0}^ja_i=c_0+(c_1-c_0)+(c_2-c_1)+(c_3-c_2)+\cdots+(c_k-c_{k-1})=c_k$. Thus condition (iii) of the lemma is precisely the statement that all $c_i$ are non-negative.

Lemma  \ref{lemma:KaNo63} is itself a special case of a more general result from the same paper.

\begin{lemma}\cite{KaNo63} \label{lemma:KaNo63Cont} Let $\mu$ be a signed measure on $[a,b]$. Let $\mu_1(x)=\int_a^x\,d\mu$ and $\mu_2(x)=\int_a^x\mu_1(t)\,dt$.
Then the inequality $\int_a^b\phi\, d\mu\geq 0$ is satisfied for all convex functions $\phi$  defined on $[a, b]$ if and only if the following three conditions hold:
\begin{enumerate}
\item[(i)] $\int_a^b\,d\mu=0$;
\item[(ii)] $\int_a^bx\,d\mu=0$;
\item[(iii)] $\mu_2(x)\geq 0\quad \forall x\in [a,b]$.
\end{enumerate}
\end{lemma}

By considering any measure $\mu$ with support strictly on some set of integers $\{i\}$ having associated point masses $a_i$, we then obtain the discrete analogue---lemma \ref{lemma:KaNo63} considered above.  Lemma \ref{lemma:KaNo63Cont} can also be used to characterize finite sequences $\{ a_i\}_{i\in I}$ where $I$ is a finite set of real numbers (and not necessarily integers) such that $\sum_{i\in I}a_i\phi(i)\geq 0$ for all convex functions $I$ by considering the signed measure $\mu$ supported on $I$ with $\mu(i)=a_i$.

\subsection{Trumping}\label{sec:trumping}

\begin{definition}
Let $x, y \in \mathbb{R}^d$ of non-negative components. We say that $x$ is \emph{trumped} by $y$, written $x\prec_T y$, if there exists a vector $c\in \R^n$ with positive components such that $x\otimes c\prec y\otimes c$. 
\end{definition}
The vector $c$ is often called a \emph{catalyst}.

Let us define $\sigma(x)=-\sum_{i=1}^dx_i\log x_i$, which we recall is the formula for the von Neumann entropy of a state with eigenvalues $x_i$. von Neumann entropy measures the amount of uncertainty we have regarding the physical state of a quantum  system. Let us also define $A_\nu(x)=\left(\frac1d\sum_{i=1}^dx_i^\nu\right)^{\frac1\nu}$ for real numbers $\nu\neq 0$ and $A_0(x)=\left(\prod_{i=1}^d x_i\right)^\frac{1}{d}$. In \cite{Tur07}, Turgut established the following result.

\begin{theorem}\cite{Tur07} \label{turg}
For two real $d$-dimensional vectors $x$ and $y$ with non-negative components such that $x$ has non-zero elements and the vectors are distinct (i.e.\ $x^{\downarrow}\neq y^{\downarrow}$), the relation $x\prec_T y$ is equivalent to the following three strict inequalities:
\begin{enumerate}
\item[(T1)] $A_\nu(x)>A_\nu(y),\quad \forall \nu\in (-\infty, 1)$,\label{T1}
\item[(T2)] $A_\nu(x)<A_\nu(y),\quad \forall \nu\in (1,\infty)$,\label{T2}
\item[(T3)] $\sigma(x)>\sigma(y)$.
\end{enumerate}
\end{theorem}

In \cite{Kli2004,Kli2007}, Klimesh stated a similar equivalent condition for trumping. For a $d$-dimensional probability vector $x$, let
\[
  f_r(x) =
  \begin{cases}
    \ln \sum_{i=1}^d x_i^r & (r>1); \\
    \sum_{i=1}^d x_i \ln x_i & (r = 1); \\
    -\ln \sum_{i=1}^d x_i^r & (0<r<1); \\
    -\sum_{i=1}^d \ln x_i & (r = 0); \\
    \ln \sum_{i=1}^d x_i^r & (r<0).
  \end{cases}
\]

If any of the components of $x$ are
$0$, we take $f_r(x) = \infty$ for $r \leq 0$.

\begin{theorem}(\cite{Kli2004,Kli2007})\label{thm:Klimesh}
 Let $x=(x_1,\ldots,x_d)$ and $y=(y_1,\ldots,y_d)$ be $d$-dimensional
 probability vectors.
 Suppose that $x$ and $y$ do not both contain components equal to $0$
 and that $x^{\downarrow} \neq y^{\downarrow}$.  Then $x \prec_T y$ if and only if
 $f_r(x) < f_r(y)$ for all real numbers $r$.
\end{theorem}

Klimesh and Turgut's conditions are easily  seen to be equivalent.

\subsection{Dirichlet Polynomials}
\begin{definition}
A \emph{general Dirichlet polynomial} is a polynomial of the form
\[
    \sum_{n=1}^{k}a_n e^{-\lambda_n s},
\]
where, herein, we take $a_n, s\in \R$ (in general, they can be complex) and $\{\lambda_n\}$ is a strictly increasing sequence of positive numbers that tends to infinity.

We obtain a \emph{Dirichlet polynomial}
\[
    \sum_{n=1}^{k}\frac{a_n}{n^s},
\]
via $ \lambda_n=\log n$.

Our results make use of general Dirichlet polynomials; we do not address the case when the number of summands is infinite.
\end{definition}

\subsection{Completely Monotone Functions}

\begin{definition}
Let $I$ be a real interval.  A function $f$ is said to be \emph{completely monotone} on $I$ if $(-1)^nf^{(n)}(x)\geq 0$ for all $x\in I$ and all $n=0, 1, 2, \dots$.
\end{definition}

The following elementary observation will be useful to us later on:

\begin{lemma} \label{easy} Any non-zero entire function that is completely monotone on $(0,\infty)$ must be strictly positive on $\mathbb{R}$. \end{lemma}

\begin{proof} Let $f$ be a non-zero entire function that is completely monotone on $(0,\infty)$.  Then $f$ is non-negative and non-increasing on $[0,\infty)$.  If there exists $c\in [0,\infty)$ such that $f(c)=0$ then $f$ must be zero on $[c,\infty)$; since the zero set of a non-zero entire function cannot have a limit point, this is impossible and $f$ is strictly positive on $[0,\infty)$.  Since the derivatives of $f$ are continuous of all orders,  $(-1)^nf^{(n)}(0)\geq 0$ and a MacLaurin series argument shows us that $f$ is strictly positive on $(-\infty,0)$ as well.\end{proof}

The example $f(x)=\frac{1}{1+x}$ shows us that the requirement that $f$ be entire in the above lemma cannot be removed.
Completely monotone functions are necessarily positive, decreasing, and convex. Bernstein's theorem on monotone functions characterizes completely monotone functions $f$ on $(0, \infty)$ via
\[
f(s) = \int_0^\infty e^{-st} \,d\mu(t),
\]
the Laplace transform of a positive measure.
In this way the completely monotone functions on $(0, \infty)$ are seen to be the cone generated by $e^{-st}$.

We recall that the Mellin transform of a function $f$ on $(0,\infty)$ is the function $\phi(s)=\int_{0}^{\infty}f(t)t^{s-1}dt$.  The Mellin and Laplace transforms are closely related.  One such connection is the following observation:

\begin{observation} Let $f\in L^1((0,\infty))$ be zero outside of $[0,1]$.  Then the Mellin transform of $f(x)$ is the Laplace transform of $f(e^{-x})$. \end{observation}

This observation leads to the following useful characterization of complete positivity of Mellin transforms of functions supported on $[0,1]$, which we use in subsequent sections; in particular, for the proof of theorem \ref{thm:mainresult}.

\begin{corollary} Let $f\in L^1((0,\infty))$ be zero outside of $[0,1]$.  Then the Mellin transform of $f$ is completely monotone on $(0,\infty)$ iff $f$ is non-negative almost everywhere. \end{corollary}

Finally, we will require the fact that  the product of two completely monotone functions on $I$ is a completely monotone function on $I$. For the convenience of the reader, we prove this well-known fact below.

\begin{lemma}\label{lem:semi-g}
Let $I$ be a real interval. Then the product of completely monotone functions on $I$ is itself completely monotone on $I$.
\end{lemma}

\begin{proof}
We only need to show that if $f$ and  $g$ be completely monotone functions on $I$ then the product $fg$ is completely monotone on $I$. Indeed, using the product rule for higher order derivatives, we find
\begin{eqnarray*}
(-1)^n(fg)^{(n)}(x)&=&(-1)^n\sum_{k=0}^n{n\choose k} f^{(n-k)}g^{(k)}(x)\\
&=&\sum_{k=0}^n{n\choose k}(-1)^{n-k}f^{(n-k)}(x)(-1)^kg^{(k)}(x).
\end{eqnarray*}
Both $(-1)^{n-k}f^{(n-k)}(x)$ and $(-1)^kg^{(k)}(x)$ are non-negative for all $x\in I$ and all $n-k$ ($k$, respectfully) $=0, 1, 2, \dots$. 
It follows that $fg$ is completely monotone on $I$ by definition.

\end{proof}

\section{Connection Between Dirichlet polynomials and Majorization \& Trumping}\label{sec:Connection}

Consider the following Dirichlet polynomial: $\zeta(s)=\sum_i\frac1{y_i^s}-\sum_i\frac1{x_i^s}$, which we can write as $\sum_n\frac{a_n}{n^s}$ with $a_n$ as in equation (\ref{eq:a_n}). We will relax our assumption that the $x_i, y_i$ need to be integers. For non-integer values we can simply use a general Dirichlet polynomial $\sum_{n=1}^k a_n e^{-\lambda_ns}$. For ease and consistency of notation we will still use $\sum_n\frac{a_n}{n^s}$ where $n$ is summed over all real numbers in $(1,\infty)$ rather than the natural numbers.  Note that is still a finite sum since $a_n=\# \{ i:y_i=n\} -\# \{ i:x_i=n\} $.

By considering this generalized Dirichlet polynomial, the discrete case of lemma \ref{lemma:KaNo63Cont} can be re-written as
\begin{theorem}\label{lemma:KaNo63D} Let $x$ and $y$ be vectors of the same length all of whose entries lie in $(1,\infty)$ and let $\zeta(s)=\sum_i\frac1{y_i^s}-\sum_i\frac1{x_i^s}$.
Then $x\prec y$ if and only if
\begin{enumerate}
\item[(i)] $\zeta(0)=0$;
\item[(ii)] $\zeta(-1)=0$;
\item[(iii)] $\frac{\zeta(s)}{s(s+1)}$ is completely monotone on $(0, \infty)$.
\end{enumerate}
\end{theorem}

\begin{proof}
Items (i) and (ii) can easily be seen to correspond to items (i) and (ii) of lemma \ref{lemma:KaNo63Cont}. We show the correspondence of the respective items (iii).

Using a well-known result by Abel in Analytic number theory, we have
\begin{eqnarray*}
\zeta(s)=\sum_{n=0}^x\frac{a_n}{n^s}=\frac{\mu_1}{x^s}+s\int_1^x\mu_1(t)t^{-s-1}\,dt,
\end{eqnarray*}
with $x$ large enough to capture all the terms and $\mu_1(t)=\sum_{0\leq n\leq t}a_n$.

Since $\mu_1(t)=0$ for $t$ sufficiently large, the term $\frac{\mu_1}{x^s}\rightarrow 0$ as $x\rightarrow \infty$. We then have
\begin{eqnarray*}
\zeta(s)&=&s\int_1^x\mu_1(t)t^{-s-1}\,dt \\
&=&t^{-(s+1)}\mu_2(t)\,|_1^\infty+s(s+1)\int_1^\infty\mu_2(t)t^{-(s+2)}\, dt,
\end{eqnarray*}
where we have used integration by parts.  Since $\mu_2(t)=0$ for $t$ sufficiently large, the first term on the RHS tends to 0 as $t\rightarrow \infty$ (it is also equal to zero at $t=1$). We therefore obtain
\begin{eqnarray*}
\zeta(s)=s(s+1)\int_1^\infty\mu_2(t)t^{-(s+2)}\, dt.
\end{eqnarray*}

Dividing by $s(s+1)$, we obtain
\[
\frac{\zeta(s)}{s(s+1)}=\int_1^\infty\mu_2(t)t^{-(s+2)}\, dt.
\]
By making the change of variables $t=e^x$, we obtain
\begin{eqnarray}\label{eq:DirichletLaplacet}
\frac{\zeta(s)}{s(s+1)}=\int_0^\infty\mu_2(e^x)e^{-x}e^{-sx}\, dx.
\end{eqnarray}
This is precisely the Laplace transform of the function $\mu_2(e^x)e^{-x}$.

Suppose $\mu_2(\cdot)\geq 0$, so that  $\mu_2(e^x)e^{-x}\geq0$.
%
%
%
By Bernstein's theorem, the Laplace transform of a positive measure is completely monotone. A generalized Dirichlet polynomial is always the Laplace transform of a function (not simply a general measure); moreover, Laplace transforms are unique up to differences on sets of measure 0. Therefore, provided $\mu_2(\cdot)\geq 0$, we have that $\frac{\zeta(s)}{s(s+1)}$ is completely monotone on $(0, \infty)$. Furthermore, if $\frac{\zeta(s)}{s(s+1)}$ is completely monotone on $(0, \infty)$  and if equation (\ref{eq:DirichletLaplacet}) holds, again using the uniqueness of the Laplace transform as well as Bernstein's theorem, it follows that $\mu_2(\cdot)\geq0$.

Note that $\mu_2(\cdot)\geq 0$ is precisely condition (iii) of lemma \ref{lemma:KaNo63Cont}.
\end{proof}

If there exists a catalyst $c$ such that $x\otimes c\prec y\otimes c$, then  condition (iii) of theorem \ref{lemma:KaNo63D} says
that if  $\tilde{\zeta}$ is the corresponding Dirichlet polynomial then
$\frac{\tilde{\zeta}(s)}{s(s+1)}$
is completely monotone on $(0, \infty)$.
It is easy to see that $\tilde{\zeta}(s)$ is given by
\begin{eqnarray*}
\tilde{\zeta}(s)&=&\sum_{i,j}\frac1{(y_ic_j)^s}-\frac1{(x_ic_j)^s}\\
&=&\left(\sum_i\frac1{y_i^s}-\frac1{x_i^s}\right)\left(\sum_j\frac1{c_j^s}\right).
\end{eqnarray*}
Thus we can write $\tilde{\zeta}(s)=\zeta(s)\zeta_2(s)$ where $\zeta_2(s)=\sum_j\frac1{c_j^s}$ is the Dirichlet polynomial corresponding to the non-zero catalyst vector $c$.

This leads us to the following observation:

\begin{observation} \label{Obs:Trump} Let $x$ and $y$ be vectors of the same length all of whose entries lie in $(1,\infty)$ and let $\zeta(s)=\sum_i\frac1{y_i^s}-\sum_i\frac1{x_i^s}$. By theorem \ref{lemma:KaNo63D} and the definition of trumping, we have $x\prec_T y$ if and only if 
\begin{enumerate}
\item[(i)] $\zeta(0)=0$;
\item[(ii)] $\zeta(-1)=0$;
\item[(iii)] There exists a Dirichlet polynomial $\zeta_2\not\equiv 0$ with non-negative coefficients such that
\[
\frac{\zeta(s)\zeta_2(s)}{s(s+1)}
\]
is completely monotone  on $(0, \infty)$.
\end{enumerate}
\end{observation}

We can reformulate the conditions (T1)-(T3) of Turgut's theorem by using the general Dirichlet polynomial $\zeta(s)=\sum_na_ne^{-\lambda_ns}$.  We note that, in terms of the functionals $f_r$ from theorem \ref{thm:Klimesh}, we have $\zeta(s)=\exp\{f_{-s}(y)-f_{-s}(x)\}$ when $s\not \in (-1,0)$; $\zeta(s)=\exp\{f_{-s}(x)-f_{-s}(y)\}$ when $s\in (-1,0)$; $\zeta^\prime (0)=f_{0}(y)-f_{0}(x)$ and $\zeta^\prime (-1)=f_{1}(x)-f_{1}(y)$.  We can use these observations to rewrite Turgut's theorem in terms of Dirichlet polynomials.

The following proposition can be taken as a succinct restatement of Turgut's theorem (theorem \ref{turg}).

\begin{proposition} Let $x$ and $y$ be vectors of the same length all of whose entries lie in $(1,\infty)$ and let $\zeta(s)=\sum_i\frac1{y_i^s}-\sum_i\frac1{x_i^s}$.  Then the following statements are equivalent:

\begin{enumerate}
\item $x\prec_T y$;
\item $\zeta$ is a general Dirichlet polynomial with simple zeros at $-1$ and $0$ such that
\[
\frac{\zeta(s)}{s(s+1)}
\]
is positive for all $s\in \mathbb{R}$;
\end{enumerate}
\end{proposition}

The proof follows from theorem \ref{thm:mainresult}.

If we do not assume $\frac{\zeta(s)}{s(s+1)}$ is completely monotone on $(0,\infty)$, the integral $\int_0^\infty\mu_2(e^x)e^{-x}e^{-sx}\, dx$ could be positive \emph{without} $\mu_2$ being positive. For this case, it is possible that $\frac{\zeta(s)}{s(s+1)}>0$ without $\frac{\zeta(s)}{s(s+1)}$ being completely monotone on $(0, \infty)$. This would be precisely the case of non-trivial trumping (that is, trumping that is not majorization).

\section{Higher Order Convexity}\label{sec:r-convex}
Consider a sequence $\chi=(\chi_n)_n$. For any natural number $r$, we define the \emph{$r$-th difference operator} \cite{PaSi58} by $\Delta^r\chi_n=\sum_{j=0}^r(-1)^j{r \choose j}\chi_{n-j}$, with $\chi_k\equiv 0$ for negative $k$. Thus, for example, $\Delta\chi_n=\chi_n-\chi_{n-1}$ and $\Delta^2\chi_n=\chi_{n}-2\chi_{n-1}+\chi_{n-2}$. We say that a sequence $\chi=(\chi_n)_n$ is \emph{$r$-convex} if $\Delta^r\chi_n\geq 0$ for all $n=1, 2,\dots$.

We can also define higher order continuous functions using divided differences.

\begin{definition} Let $I$ be an interval and $f:I\to \R$ and let $\{ x_j\}$ be distinct elements of $I$.  Then we define the first order divided difference as $f[x_0,x_1]=\frac{f(x_1)-f(x_0)}{x_1-x_0}$.  The second order divided difference is $f[x_0,x_1,x_2]=\frac{f[x_1,x_2]-f[x_1,x_0]}{x_2-x_0}$.  Higher order divided differences are defined in a similar inductive manner $f[x_0,...,x_n]=\frac{f[x_1,...,x_n]-f[x_{n-1},...,x_0]}{x_n-x_0}$.  An equivalent characterization is that $f[x_0,...,x_n]$ is the coefficient of $x^n$ in the Lagrange interpolating polynomial to $f$ at the nodes $x_0,x_1,...,x_n$. \end{definition}

\begin{definition}\cite{Bul71} Let $I$ be an interval and $f:I\to \R$ and let $r\in \mathbb{N}$, $f$ is said to be a \emph{convex function} of order $r$ (or $r$-convex function) on $I$ if $f[x_0,...,x_r]\geq 0$ whenever $\{x_j\}_{j=0}^{r}$ is a $(r+1)$-tuple of distinct numbers in $I$. \end{definition}

We note that  the term ``1-convex'' is equivalent to ``increasing'', and ``2-convex'' is equivalent to ``convex'', for both higher order convex sequences and higher order convex functions.

\begin{observation} There is a close relationship between $r$-convex functions and $r$-convex sequences.  If $f$ is an $r$-convex function on $(0,\infty)$, then $\{ f(n) \}_{n\in \mathbb{N}}$ is an $r$-convex sequence.
\end{observation}

We define  $\mu_k(x)=\int_a^x\mu_{k-1}(t)\,dt$ for $k=1, 2, 3, \dots$, where $\mu_0:=\mu$.  The $r$-convex analogues to lemmas \ref{lemma:KaNo63Cont} and \ref{lemma:KaNo63} are as follows.

\begin{lemma} Let $r\in \mathbb{N}$.
The inequality
\begin{eqnarray}\label{eq:KNrconvCont}
\int_a^b\phi(t)\,d\mu\geq 0
\quad\textnormal{ for all $r$-convex functions $\phi$}
\end{eqnarray}
  is equivalent to the following two conditions.
\begin{enumerate}
\item[(i)] $\int_a^bx^k\,d\mu=0$ for $k=0,1,...,r-1$;
\item[(ii)] $(-1)^r\mu_r(x)\geq 0$ for all $x\in [a,b]$.
\end{enumerate}
\end{lemma}

\begin{proof}
The proof follows the original proof from \cite{KaNo63}. If the inequality (\ref{eq:KNrconvCont}) holds for all $r$-convex functions $\phi$, then since $\psi_\pm:=\pm x^k$ are $r$-convex for $k=0,1,...,r-1$, item (i) follows immediately. The $k=0$ case of item (i) implies $\mu_1(b)=0$, and the $k=0$ and $k=1$ cases of item (i) together imply $\mu_2(b)=0$, these cases together with the $k=2$ case in turn implies $\mu_3(b)=0$. We can repeat this to obtain $\mu_k(b)=0$, for all $1\leq k\leq r$. Applying integration by parts to $\int_a^b\phi(t)\,d\mu$  $r$ times yields
\begin{eqnarray}\label{eq:intmu_r}
\int_a^b\phi(t)\,d\mu=(-1)^r\int_a^b\phi^{(r)}\mu_r(x)\,dx.
\end{eqnarray}
The assumption that equation (\ref{eq:intmu_r}) is non-negative, together with the assumption that $\phi$ is $r$-convex, gives item (ii).

These two items are also sufficient: For an $r$-th differentiable function $\phi$ satisfying  items (i) and (ii), equation (\ref{eq:intmu_r}) holds, which, assuming $\phi$ is $r$-convex, implies $\int_a^b\phi(t)\,d\mu\geq 0$.
\end{proof}

The discrete version of the above lemma is stated below for completeness.

\begin{lemma}\label{lemma:KNrconvDiscrete} Let $r\in \mathbb{N}$.
The inequality
\begin{eqnarray}\label{eq:KNrconv}
\sum_{0\leq i\leq m} a_i\chi_n\geq 0
\quad\textnormal{ for all $r$-convex sequences $(\chi_n)_n$}
\end{eqnarray}
  is equivalent to the following two conditions.
\begin{enumerate}
\item[(i)] $\sum_{0\leq n\leq m}n^k a_n=0$ for $k=0,1,...,r-1$;
\item[(ii)] If $r$ is even, $\sum_{i_{r}=0}^{i_r+1}\sum_{i_{r-1}=0}^{i_r}\dots\sum_{i_2=0}^{i_3}\sum_{i_1=0}^{i_2}a_{i_1}\geq 0$ where $0\leq i_j\leq m$ for all $j\leq r+1$, with the reverse inequality for odd $r$.
\end{enumerate}
\end{lemma}
In the discrete setting, it is notationally convenient to  denote  $\sum_{i_{r}=0}^{i_r+1}\sum_{i_{r-1}=0}^{i_r}\dots\sum_{i_2=0}^{i_3}\sum_{i_1=0}^{i_2}a_{i_1}$ by $\mu_r$, as in the continuous case.

\medskip

Using Descartes' rule of signs, we can restate Corollary 4.1 of \cite{Nie05} ever so slightly, as follows.
\begin{theorem}
Let $a=(a_0, a_1, \dots)$ be a real vector. The inequality $\sum_na_n\chi_n\geq 0$ holds for all $r$-convex sequences $\chi=(\chi_n)_n$ if an only if the polynomial $p(z)=\sum_na_nz^n$ has a root of multiplicity $r$ at $z=1$ and all the coefficients of the polynomial  $\frac{p(z)}{(z-1)^r}$ are non-negative.
\end{theorem}

For our setting, we wish to prove a generalized version of this theorem for Dirichlet polynomials.
\begin{theorem}\label{thm:DSrconv}
Let $a=(a_0, a_1, \dots)$ be a real vector. The inequality $\sum_na_n\chi_n\geq 0$ holds for all $r$-convex sequences $\chi=(\chi_n)$ if an only if the Dirichlet polynomials $\zeta(s)=\sum_n\frac{a_n}{n^s}$ has zeros at $s=0, -1, -2, \dots, -r+1$  and $\frac{(-1)^r\zeta(s)}{\Pi_{k=0}^{r-1}(s+k)}$ is completely monotone on $(0, \infty)$.
\end{theorem}

\begin{proof}
 The bulk of the proof follows the calculations  preceeding theorem \ref{lemma:KaNo63D}, which use a well-known result due to Abel, which we restate here for convenience.

 \begin{eqnarray*}
\zeta(s)=\sum_{n=0}^x\frac{a_n}{n^s}=\frac{\mu_1}{x^s}+s\int_1^x\mu_1(t)t^{-s-1}\,dt,
\end{eqnarray*}
with $x$ large enough to capture all the terms.

If $s>0$, we note that the first term on the RHS vanishes. We can integrate the second term on the RHS by parts, and the first term of the result vanishes as well. Continuing in this manner, integrating by parts $r-1$ times (where $r\geq 2\in \mathbb{Z}$) yields
 \begin{eqnarray*}
\zeta(s)=\Pi_{k=0}^{r-1}(s+k)\int_1^\infty\mu_{r}(t)t^{-(s+r)}\,dt.
\end{eqnarray*}
Dividing by $\Pi_{k=0}^{r-1}(s+k)$, multiplying both sides of the equation by $(-1)^r$,  and making the change of variables $t=e^x$, we obtain
 \begin{eqnarray}\label{eq:cm=LT}
\frac{(-1)^r\zeta(s)}{\Pi_{k=0}^{r-1}(s+k)}=\int_0^\infty(-1)^r\mu_{r}(e^x)e^{-rx}e^{-sx}\,dx.
\end{eqnarray}
It follows from  lemma \ref{lemma:KNrconvDiscrete} that  $(-1)^r\mu_{r}(e^x)e^{-rx}\geq 0$ for any $r$. Similar to before, we note that equation (\ref{eq:cm=LT}) is precisely the Laplace transform of the function $(-1)^r\mu_{r}(e^x)e^{-rx}$, and we conclude  $\frac{(-1)^r\zeta(s)}{\Pi_{k=0}^{r-1}(s+k)}$ is a completely monotone function on $(0,\infty)$.

For the reverse direction, looking at  equation (\ref{eq:cm=LT}) and noting Bernstein's theorem, the result follows from lemma \ref{lemma:KNrconvDiscrete}.
\end{proof}

Note that theorem \ref{thm:DSrconv} is a direct generalization of theorem \ref{lemma:KaNo63D} with the same setup as before, namely $a_n=\# \{ i:y_i=n\} -\# \{ i:x_i=n\} $ for all $n\in \mathbb{N}$ (and using a general Dirichlet polynomial in the case of non-integer values). The inequality $\sum_na_n\chi_n\geq 0$ holding for all $r$-convex sequences (rather than for all convex sequences, as is the case for $x\prec y$) is a more generalized partial order on real vectors.

\medskip

We note that we can write $\frac{\zeta(s)}{\prod_{k=0}^{r-1}(s+k)}$ as a Mellin transform.  We remind the reader of the truncated power notation used in the theory of splines $(1-nx)^{r-1}_{+}=(\max(1-nx,0))^{r-1}$.
We are now ready to state our result:

\begin{proposition} Let $\zeta(s):=\sum_{n=1}^{\infty} \frac{a_n}{n^s}$.  Then $\frac{\zeta(s)}{\prod_{k=0}^{r-1}(s+k)}$ is the Mellin transform of $\frac{1}{(r-1)!}\sum_{n=1}^{\infty} a_n(1-nx)^{r-1}_{+}$.   \end{proposition}

\begin{proof}
The proof is a straightforward calculation using integration by parts.

To simplify the manipulations of expressions, let $g_{q,r}$ be the Mellin transform of $x^q(1-nx)_+^{r-1}$. Note that the expression $\big((1-nx)_+\big)^0$ is $1$ for $x\in [0, 1/n]$ and $0$ for $x>1/n$. This yields
\begin{eqnarray*}
g_{q,0}&=&\int_0^\infty x^q\chi_{[0,1/n]}(x)x^{s-1}\,dx\\
&=&\int_0^{1/n}x^{q+s-1}\,dx\\
&=&\frac{n^{-(q+s)}}{q+s}.
\end{eqnarray*}
For $r-1>0$, we use integration by parts with $u=x^q(1-nx)^{r-1}$ and $dv=x^{s-1}dx$ so that
\begin{eqnarray*}
g_{q,r}&=&\int_0^\infty x^q(1-nx)_+^{r-1}(x)x^{s-1}\,dx\\
&=&(1-nx)^{r-1}\frac{x^{q+s}}{q+s}\big|_0^{1/n}+\frac{n(r-1)}{q+s}\int_0^{1/n}(1-nx)^{r-2}x^{q+s}\,dx\\
&=&\frac{n(r-1)}{q+s}\int_0^{1/n}(1-nx)^{r-2}x^{q+s}\,dx\\
&=& \frac{n(r-1)}{q+s}g_{q+1, r-1}.
\end{eqnarray*}

In general, we obtain
\[
g_{q,r}=\frac{n^{-(q+s)}(r-1)!}{(q+s)(q+s+1)\cdots(q+s+r-1)}.
\]

Letting $q=0$, we have
\[
g_{0,r}=\frac{n^{-s}}{\prod_{k=0}^{r-1}(s+k)}(r-1)!
\]

Since  the Mellin transform of $(1-nx)_+^{r-1}$ is $g_{0,r}$, we see that
the Mellin transform of $\frac{1}{(r-1)!}\sum_{n=1}^{\infty} a_n(1-nx)^{r-1}_{+}$ is
\begin{eqnarray*}
\frac{1}{(r-1)!}\sum_{n=1}^\infty a_n g_{0,r}&=&\frac{\sum_{n=1}^{\infty} a_nn^{-s}}{\prod_{k=0}^{r-1}(s+k)}\\
&=&\frac{\zeta(s)}{\prod_{k=0}^{r-1}(s+k)},
\end{eqnarray*} as desired.
\end{proof}

We now prove a generalization of Turgut's theorem.  We first state a lemma from \cite{Tur07}.  This lemma is a special case of a result of P\'{o}lya \cite{Pol28}.  This result was rediscovered in \cite{Tur07}.

\begin{lemma}\label{lem:Pol} Let $f(x)$ be a real polynomial which is positive on $(0,\infty)$.  Then there exists two real polynomials $g(x)$ and $h(x)$ with all coefficient nonnegative such that $f(x)g(x)=h(x)$.  \end{lemma}

\begin{theorem} \label{thm:mainresult} Let $\zeta(s)$ be a generalized Dirichlet polynomial with simple zeros at $s=0,-1,...,-(r-2),-(r-1)$.  Then $\frac{\zeta(s)}{\prod_{k=0}^{r-1}(s+k)}$ is positive on the real line if and only if there exists a generalized Dirichlet polynomial $\zeta_2(s)\not\equiv 0$ with non-negative coefficients such that $\frac{\zeta(s)\zeta_2(s)}{\prod_{k=0}^{r-1}(s+k)}$ is completely monotone on $(0,\infty )$. \end{theorem}

\begin{proof}  Suppose $g(s)=\frac{\zeta(s)\zeta_2(s)}{\prod_{k=0}^{r-1}(s+k)}$ is completely monotone on $(0,\infty )$.  By lemma \ref{easy}, $g(s)$ must be positive on the real line.  Since $\zeta_2(s)$ is strictly positive on the real line, $\frac{g(s)}{\zeta_2(s)}=\frac{\zeta(s)\zeta_2(s)}{\prod_{k=0}^{r-1}(s+k)}$ must be strictly positive on the real line.

The proof of other direction is modeled after the proof of the main result in \cite{Tur07}. As in \cite{Tur07}, we begin with the special case that $\zeta(s)=\sum_{n=0}^k a_n e^{-\lambda_ns}$ where $\lambda_n=n\alpha$ where $\alpha$ is some positive number.  Let $p(x)=\sum_{n=0}^{k}a_nx^n$, then $\zeta(s)=p(e^{-\alpha s})$.  If the only real zeros of $\zeta(s)$ are simple zeros at $-(r-1),-(r-2),...,-1,0$ then the only positive real zeros of $p(x)$ are simple zeros at $1,e^{\alpha},e^{2\alpha},...,e^{(r-1)\alpha}$.  Hence $p(x)=f(x)\prod_{k=0}^{r-1}(e^{k\alpha}-x)$ where $f(x)$ is a real polynomial which is positive on $(0,\infty)$.  Hence, by lemma \ref{lem:Pol},  there exists two real polynomials $g(x)$ and $h(x)$ with all coefficients nonnegative such that $f(x)g(x)=h(x)$.  Let $\zeta_2(s)=g(e^{-\alpha s})$, then $\zeta_2(s)$ is a generalized Dirichlet polynomial with nonnegative coefficients. Note that $\frac{\zeta(s)\zeta_2(s)}{\prod_{k=0}^{r-1}(s+k)}=h(e^{-\alpha s})\prod_{k=0}^{r-1}\frac{e^{\alpha k}-e^{-\alpha s} }{s+k}$.  Since $h(e^{-\alpha s})$ is a generalized Dirichlet polynomial with nonnegative coefficients it is completely monotone on $(0,\infty)$.  The functions  $e^{\alpha k}-e^{-\alpha s}$ and $\frac{1}{s+k}$ are completely monotone on $(0,\infty)$ for all $k\in \mathbb{N}\cup \{ 0\}$.  Therefore the product $\frac{\zeta(s)\zeta_2(s)}{\prod_{k=0}^{r-1}(s+k)}$ is completely monotone on $(0,\infty )$ by lemma \ref{lem:semi-g}.  The general case of this result now follows from an approximation argument similar to that of \cite{Tur07}.

\end{proof}

Theorem \ref{turg}, which is the main result of \cite{Tur07}, can easily be seen as a corollary: when $r=2$, the statement of our theorem is as follows:

\begin{corollary}
 Let $\zeta(s)$ be a generalized Dirichlet polynomial with simple zeros at $s=0,-1$.  Then $\frac{\zeta(s)}{s(s+1)}$ is positive on the real line if and only if there exists a generalized Dirichlet polynomial $\zeta_2(s)\not\equiv 0$ with non-negative coefficients such that $\frac{\zeta(s)\zeta_2(s)}{s(s+1)}$ is completely monotone on $(0,\infty )$.
 \end{corollary}

 This is exactly observation \ref{Obs:Trump}.

\section{Conclusion}

We have found a characterization of the majorization relation between two vectors in terms of the complete montonicity of an associated Dirichlet polynomial. Since the majorization relation describes the transformations of bipartite states by LOCC, this result may prove useful in future studies of such transformations.  Furthermore, since Dirichlet series and completely monotone functions have been studied extensively by mathematicians, there are potentially a large number of mathematical results which may prove useful in this area.  An example can be found in this paper where these ideas were used to simplify the derivation of Turgut's theorem and give a generalization of this result to higher order convex functions.  While the generalization to higher ordered convex functions does not currently have a physical application, the simplification of Turgut's proof shows that the techniques used in this paper are useful in the analysis of entanglement transformations and catalysis.

\section*{Acknowledgements}  We thank the referees for many helpful suggestions which improved this paper.  R.P. was supported by NSERC Discovery Grant 400550. S.P. was supported by an NSERC Doctoral Scholarship.

\end{document}